\documentclass[11pt]{article}
\usepackage{amsmath}
\usepackage{amsthm}
\usepackage{amsfonts}
\usepackage[latin5]{inputenc}

\numberwithin{equation}{section}
\newcounter{thmcounter}
\newcounter{Remarkcounter}
\newcounter{ExampleCounter}
\numberwithin{thmcounter}{section}

\newtheorem{Prop}[thmcounter]{Proposition}
\newtheorem{Corol}[thmcounter]{Corollary}
\newtheorem{theorem}[thmcounter]{Theorem}
\newtheorem{Lemma}[thmcounter]{Lemma}
\theoremstyle{remark}
\newtheorem{Remark}[Remarkcounter]{Remark}
\newtheorem{Example}[ExampleCounter]{Example}
\begin{document}

\title{On the marginally trapped surfaces in 4-dimensional space-times with finite type Gauss map}
\author{Nurettin Cenk Turgay\footnote{The final version of this paper was published in General Relativity and Gravitation, see \cite{NCTGenRelGrav}} \footnote{ e-mail:turgayn@itu.edu.tr, Adress:Istanbul Technical University, Faculty of Science and Letters,
Department of  Mathematics, 34469 Maslak, Istanbul, Turkey}}

\date{\today}
\maketitle
\begin{abstract}
In this paper, we work on the marginally trapped surfaces in the 4-dimensional Minkowski, de Sitter and anti-de Sitter space-times. We obtain the complete classification of the marginally trapped surfaces in the Minkowski space-time with pointwise 1-type Gauss map.   Further, we give construction of marginally trapped surfaces with 1-type Gauss map and a given boundary curve. We also state some explicit examples. We also prove that a marginally trapped surface in the de Sitter space-time $\mathbb S^4_1(1)$ or anti-de Sitter space-time $\mathbb H^4_1(-1)$  has pointwise 1-type Gauss map if and only if its mean curvature vector is parallel. Moreover, we obtain that there exists no marginally trapped surface in $\mathbb S^4_1(1)$ or $\mathbb H^4_1(-1)$ with harmonic Gauss map.

\textbf{Keywords}: Minkowski space-time, marginally trapped surface,  finite type Gauss map, null 2-type, de Sitter space-time

\textbf{Mathematics Subject Classification}: 53B25, 53C40
\end{abstract}

\section{Introductions}

Let $M$ be an $n$-dimensional semi-Riemannian submanifold of a semi-Euclidean space and $\phi$ a mapping defined on $M$ into another semi-Euclidean space. If $\phi$ is one of geometrically important mappings on $M$ such as its Gauss map or its position vector, then understanding spectrum of $\phi$ with respect to the Laplace operator $\Delta$ of $M$ may play great role to understand geometry of $M$. In this direction, the notion of finite type mappings was introduced by B. Y. Chen in late 1970's.  By the definition, if the mapping $\phi$ can be expressed as a sum of eigenvectors corresponding from $k$ distinct eigenvalues of  $\Delta$, then it is said to be  of $k$-type. If one of these eigenvalues is zero, then $M$ is called  null $k$-type, \cite{ChenKitap,ChenMakale1986}. Many important results about finite type mappings defined on semi-Riemannian submanifolds have appeared sofar \cite{Bleecker-Weiner,ChenRapor,Chen-Morvan-Nore,Chen-Piccinni,ReillyHelvaci}.

In particular, the Gauss map of submanifolds has been worked in several articles in this direction after  some results on the submanifolds with 1-type Gauss map or 2-type Gauss map had been given in \cite{Chen-Piccinni}. The Gauss map  $\nu$ of the submanifold $M$ is 1-type if and only if it satisfies 
\begin{equation}\label{Glbl1TypeDefinition}
 \Delta \nu  = \lambda (\nu+C)
\end{equation}
for some $\lambda \in \mathbb
R$ and  some constant vector $C$. However, it is well-known that the Laplacian of the Gauss map of some surfaces such as helicoids  and catenoids in the 3-dimensional Minkowski space-time and generalized torus in the 4-dimensional Euclidean space takes the form of
\begin{equation}\label{PW1TypeDefinition}
 \Delta \nu =f(\nu +C)
\end{equation}
for some smooth function $f$ on $M$ and some constant vector $C$. A submanifold with the Gauss map satisfying \eqref{PW1TypeDefinition} is said to have pointwise 1-type Gauss map (cf. \cite{KKKM,UDur,UDur2,Yoon-2}).

On the other hand,  marginally trapped surfaces play an important role in the concept of trapped surfaces, introduced by Penrose in \cite{Penrose1965}. Let $M$ be a space-like surface in a 4-dimensional space-time and $H$ its mean curvature vector. In the aspect of general relativity, $H$ measures the tension of the surface coming from the surrounding space. $M$ is said to be  trapped if $H$ is time-like on $M$.  A trapped surface has physical property that  the outgoing light rays from it are convergent. It is believed that there is a surface, called marginally trapped (or quasi-minimal) surface, to which all the light rays are parallel, \cite{CabrerizoComplete,ChenVeken2009RelNull,HaesenOrtega,Senovilla2002}.

Geometrically, a space-like surface in a 4-dimensional space-time is called marginally trapped if its mean curvature vector is light-like. In the very recent past, many problems in geometry involving marginally trapped surfaces in a space-time have been dealed with (cf. \cite{CabrerizoComplete,ChenIshikawa1991Bih,ChenVeken2009RelNull,ChenVeken2009Houston,HaesenOrtega}). For example, Haesen and Ortega obtained the classification of marginally trapped surfaces which are invariant under a group of boost isometries in \cite{HaesenOrtega}. Recently,  in \cite{CabrerizoComplete}, isotropic marginally trapped surfaces were classified.  

In this article, we study on the Gauss map of marginally trapped surfaces in the Minkowski, de Sitter and anti-de Sitter space-times. In Section 2, after we describe our notations, which are along the lines of notations used in \cite{ChenVeken2009RelNull}, we give a summary of the basic facts and formulas that we will use. In Section 3, we obtain some lemmas that we use in the other sections. In Section 4, we study on marginally trapped surfaces in the Minkowski space-time with pointwise 1-type Gauss map and give complete classification of such surfaces. Further, we give some methods for construction of marginally trapped surfaces with 1-type Gauss map. We also give some explicit examples. In Section 5, we focus on marginally trapped surfaces in the de Sitter and anti-de Sitter space-times in terms of their Gauss map. We proved that if the Gauss map of a marginally trapped surface in these space-times satisfies \eqref{PW1TypeDefinition}, then $C$ must be zero. We also proved that there is no marginally trapped surface in these space-times with harmonic Gauss map.

The surfaces we are dealing with are smooth and connected unless otherwise stated.
\section{Preliminaries}\label{SectionMinkPreliminaries}

\subsection{Basic notation, formulas and definitions}\label{SubSectBasicNot}
Let $\mathbb E^m_s$ denote the pseudo-Euclidean $m$-space with the canonical 
pseudo-Euclidean metric tensor $g$ of index $s$ given by  
$$
 g=-\sum\limits_{i=1}^s dx_i^2+\sum\limits_{j=s+1}^m dx_j^2,
$$
where $(x_1, x_2, \hdots, x_m)$  is a rectangular coordinate system in $\mathbb E^m_s$.

A non-zero vector $\zeta\in T_p(\mathbb E^m_s)\cong \mathbb E^m_s$  is called space-like (resp. time-like or light-like) if $\langle \zeta,\zeta\rangle>0$  (resp. $\langle \zeta,\zeta\rangle<0$ or $\langle \zeta,\zeta\rangle=0$), where $T_p(\mathbb E^m_s)$ denotes the tangent space of $\mathbb E^m_s$ at $p$.  We put 
\begin{eqnarray} 
\mathbb S^{m-1}_s(r^2)&=&\{x\in\mathbb E^m_s: \langle x, x \rangle=r^{-2}\},\notag
\\  
\mathbb H^{m-1}_{s-1}(-r^2)&=&\{x\in\mathbb E^m_s: \langle x, x \rangle=-r^{-2}\},\notag
\end{eqnarray}
where $\langle\ ,\ \rangle$ is the indefinite inner product of $\mathbb E^m_s$. In general relativity, $\mathbb E^4_1$, $\mathbb S^{4}_1 (r^2)$ and  $\mathbb H^{4}_1 (r^2)$
are known as the Minkowski, de Sitter and anti-de Sitter space-times, respectively, \cite{ChenVeken2009Houston}. These complete Lorentzian manifolds, which have constant sectional curvatures, are called Lorentzian space forms. We use the following notation
$$R^m_s(c)=\left\{\begin{array}{rcl} \mathbb S^{m}_s(c) &\mbox{if}& c>0\\ \mathbb E^{m}_s &\mbox{if}& c>0\\\mathbb H^{m}_s (c) &\mbox{if}& c<0 \end{array}\right.,\quad E(m,s,c)=\left\{\begin{array}{rcl} \mathbb E^{m+1}_s &\mbox{if}& c>0\\ \mathbb E^{m}_s &\mbox{if}& c>0\\\mathbb E^{m+1}_{s+1} (c) &\mbox{if}& c<0 \end{array}\right.$$
from which we see $R^m_s(c)\subset E(m,s,c)$. The light cone  $\mathcal{LC}^{m-1}$ with vertex at the origin in $\mathbb E^m_s$ is defined to be   
$
\mathcal{LC}^{m-1}  =  \{x\in\mathbb E^m_s: \langle x,x\rangle=0\}.
$

Let $M$ be an $n$-dimensional immersed semi-Riemannian submanifold  of Lorentzian space form $R^m_s(c)$. $M$ is said to be space-like if every non-zero tangent vector on $M$ is space-like. Denote Levi-Civita connections of $R^m_s(c)$ and $M$ by $\widetilde{\nabla}$ and $\nabla$,  respectively. 
In this section, we shall use letters $X,\; Y,\ Z$  to denote vectors fields tangent to $M$ and  $\xi,\; \eta$ to denote  vectors fields normal to $M$ and tangent to  $R^m_s(c)$. The Gauss and Weingarten formulas are given by
\begin{eqnarray}
\label{MEtomGauss} \widetilde\nabla_X Y&=& \nabla_X Y + h(X,Y),\\
\label{MEtomWeingarten} \widetilde\nabla_X \xi&=& -A_\xi X+D_X \xi,
\end{eqnarray}
 where $h$,  $D$  and  $A$ are second fundamental form, normal connection and  shape operator of $M$ in  $R^m_s(c)$, respectively. $A$ and $h$ of are related by
\begin{equation}
\label{MinkAhhRelatedby} \langle A_\xi X,Y\rangle=\langle h(X,Y),\xi\rangle.
\end{equation}
The mean curvature vector of $M$ in  $R^m_s(c)$  is defined by $H=\frac 1n \mathrm{tr}h$.

Denote by $R$ and $R^D$ the curvature tensor associated with the connections $\nabla$ and $D$, respectively. Then, the  Gauss, Codazzi and Ricci equations are given by
\begin{eqnarray}
\label{MinkGaussEquation} R(X,Y)Z&=&c\big( \langle Y,Z\rangle X- \langle X,Z \rangle Y\big)+ A_{h(Y,Z)}X- A_{h(X,Z)}Y,\\
\label{MinkCodazzi} (\bar \nabla_X h )(Y,Z)&=&(\bar \nabla_Y h )(X,Z),\\
\label{MinkRicciEquation} \langle R^D(X,Y)\xi,\eta\rangle&=&\langle[A_\xi,A_\eta]X,Y\rangle,
\end{eqnarray}
respectively, where  $\bar \nabla h$ is defined by
$$(\bar \nabla_X h)(Y,Z)=D_X h(Y,Z)-h(\nabla_X Y,Z)-h(Y,\nabla_X Z).$$

Let $M$ be a space-like surface in the space-time $R^4_1(c)$.   Consider  an orthonormal frame field $\{e_1,e_2;e_3,e_4\}$ on $M$. The connection forms $\omega_{AB}$ associated with this frame field are defined by $\omega_{AB}(X)=\langle\widetilde\nabla_X e_A,e_B\rangle$ and satisfy $\omega_{AB}+\omega_{BA}=0$ for $A,B=1,2,3,4$. The Gaussian curvature $K$ of  $M$ is defined by $K=\langle R(e_1,e_2)e_2,e_1\rangle$. A surface is said to be flat if $K\equiv 0$. We define normal curvature $K^D$ of $M$ in  $R^4_1(c)$ by $K^D=\langle R^D(e_1,e_2)e_3,e_4\rangle $. 

Let the mean curvature vector $H$ of  $M$ in  $R^4_1(c)$ satisfy $\langle H,H\rangle\equiv0$ on $M$. If $H$ is light-like on $M$, then $M$ is said to be a marginally trapped surface in $R^4_1(c)$. If $H$ vanishes at some points of $M$, then the surface is called partly marginally trapped, \cite{HaesenOrtega}.

Throughout this work, we denote the norm of $e_A$ by $\varepsilon_A$, i.e., $\varepsilon_A=\langle e_A,e_A\rangle=\pm 1$. In addition, the subscripts $u,v$  denotes the partial derivatives with respect to $u,v$.
\subsection{Gauss map}\label{SubSectMinkGaussMap}
Let  $\Lambda^{n}(\mathbb E^m_s)$ and $G(n, m)$ denote the space of  $n$-vectors on $\mathbb E^m_s$ and the Grassmannian manifold consisting of all $n$-planes through the origin of $\mathbb E^m_s$, respectively.  Note that $G(n, m)$  is canonically imbedded in $\Lambda^{n}(\mathbb E^m_s)$ which is an $N$ dimensional vector field, where $N= {m\choose {n}}$. A non-degenerate inner product on $\Lambda^{n}(\mathbb E^m_s)$ is defined by
$$\langle X_1\wedge X_2\wedge\cdots\wedge X_{n}, Y_1\wedge Y_2\wedge\cdots\wedge Y_{n}\rangle= \det(\langle X_i,Y_j\rangle),$$
where $X_i,\;Y_i\in\mathbb E^m_{s},\; i=1,2,\hdots,n$ and $\langle X_i,Y_j\rangle$ denotes the semi-Euclidean inner product of the vectors $X_i$ and $Y_j$. We will denote the inner product space $\Big(\Lambda^{n}(\mathbb E^m_s),\langle,\rangle\Big)$ by $\Lambda^{m,n}_S$, where $S$ is its index. There exists a one-to-one, onto and linear isometry from $\Lambda^{m,n}_S$ into $\mathbb E^N_S$, because their dimension and index are equal(see \cite[p. 52]{ONeillKitap}). Hence, we have $G(n, m)\subset\Lambda^{m,n}_S\cong\mathbb E^m_s$.  

Let $M$ be an $n$-dimensional, oriented space-like submanifold of the semi-Euclidean space  $\mathbb E^m_s$. Consider a local orthonormal base field $\{e_1, e_2,\hdots,e_{n}\}$ of the tangent bundle of $M$. Then, the Laplace operator of $M$ with respect to the induced metric is 
\begin{equation}\label{SemiEuclSpacSubmflDelta}
\Delta=\sum\limits^n_{i=1}(-e_ie_i+\nabla_{e_i}e_i).
\end{equation}
The smooth mapping
\begin{equation}\label{MinkGaussTasvTanim}
\begin{array}{rcl}\nu:M&\rightarrow&G(m-n, m)\subset S^{N-1}_S (\varepsilon)\subset \mathbb E^N_S\cong\Lambda^{m,n}_S\\
p&\mapsto&\nu(p)=(e_{1}\wedge e_{2}\wedge\hdots\wedge e_n)(p)\end{array}
\end{equation}
is called the (tangent) Gauss map of $M$  which assigns a point $p$ in $M$   to the representation of the oriented 
$n$-plane through  the origin of $\mathbb E^m_s$ and parallel  
to the tangent  space of $M$ at $p$, \cite{Chen-Piccinni,KKKM}.

A semi-Riemannian submanifold $M$ of a semi-Euclidean space $\mathbb E^m_s$ is said to have harmonic Gauss map if $\Delta\nu\equiv0$ on $M$. On the other hand, $M$ is said to have {\it pointwise 1-type Gauss map} if it satisfies \eqref{PW1TypeDefinition} for  a constant vector $C\in\mathbb E^N_S$ and a smooth fuction $f\not\equiv 0$.  Moreover, a pointwise 1-type Gauss map is called {\it of the first kind} if \eqref{PW1TypeDefinition} is satisfied for $C=0$, and  {\it of the second kind} if $C\neq 0$. Moreover, if \eqref{PW1TypeDefinition} is satisfied for a non-constant function $f$, then $M$ is said to have {\it proper} pointwise 1-type Gauss map. 

\subsection{Eigenvalues of Laplace equation in the plane}\label{subsectionHelmoltz}
 In this subsection, we use the same notation with \cite{EvansBook} and review briefly some known facts on the eigenvalues of Laplace operator. Let $\Omega$ be an open, bounded subset of $\mathbb R^2$ with the boundary $\beta=\partial\Omega$. Consider the elliptic, symmetric operator
\begin{equation}
\label{OpDeltaOmega} \begin{array}{rccc}\Delta:&H^1_0(\Omega)&\rightarrow &H^1_0(\Omega)\\
&\phi&\mapsto& -\phi_{uu}-\phi_{vv}.
\end{array}
\end{equation}
Note that a function $\phi$ is an eigenfunction of $\Delta$ given by \eqref{OpDeltaOmega} if and only if it is a solution of the boundary value problem
\begin{equation}\label{BVPHelmoltz}
\left\{\begin{array}{rrlc} \Delta\phi(u,v)+\lambda\phi(u,v)&=&0 & \mbox{if\ }(u,v)\in\Omega \\
\phi&=&0&\mbox{if\ }(u,v)\in\beta.
\end{array}\right.
\end{equation}

The following arguments hold (see \cite[p. 326-335]{EvansBook}). The eigenvalues of $\Delta$ are $0<\lambda_1\leq\lambda_2\leq\hdots\nearrow\infty$. Moreover, there exists an orthonormal basis $\{\phi_k| k=1,2\hdots\}$ of $L^2(\Omega)$ where $\phi_k\in H^1_0(\Omega)$ and satisfies
\eqref{BVPHelmoltz}
for $\lambda=\lambda_1,\lambda_2,\hdots$. The regularity theory shows that $\phi\in C^\infty(\Omega)$. Moreover, if $\beta$ is smooth, then $\phi\in C^\infty(\bar\Omega)$.


\section{Lemmas}
In this section, we obtain some lemmas that we will use in the other sections. 

Let $M$ be a space-like surface in $\mathbb E^m_s$ and $\nu$ its Gauss map. Consider an orthonormal frame field $\{e_1,e_2;e_3,e_4,\hdots,e_m\}$. From \cite[Lemma 3.2]{KKKM}, we obtain that $\nu$ satisfies 
\begin{equation}\label{Pre121Mink4GaussLaplMargTrapped}
\Delta\nu=\|\hat h\|^2\nu+\sum\limits_{3\leq\alpha\leq\beta\leq m}\varepsilon_\alpha\varepsilon_\beta\langle R^{\hat D}(e_1,e_2)e_\alpha,e_\beta\rangle e_\alpha\wedge e_\beta-2\hat D_{e_1}\hat H\wedge e_2-2e_1\wedge \hat D_{e_2}\hat H,
\end{equation}
where  $\hat D$, $\hat h$ and $\hat H$ denote normal connection, second fundemental form and mean curvature vector of $M$ in $\mathbb E^m_s$, respectively, $R^{\hat D}$ is the curvature tensor associated with $\hat D$ and $\|\hat h\|^2$ is the squared norm of $\hat h$. 

Now, consider a partly marginally trapped surface $M$ in $R^4_1(\delta)$ for $\delta\in\{1,0,-1\}$ and let $x$ be its position vector. From the equations $h(e_i,e_i)=\hat h(e_i,e_i)+\delta x$ and $h(e_i,e_j)=\hat h(e_i,e_j),\ i,j=1,2$, $i\neq j$ we obtain
\begin{equation}\label{Pre1Mink4GaussLaplMargTrapped}
\|\hat h\|^2=\langle  h(e_1,e_1), h(e_1,e_1)\rangle+2\langle  h(e_1,e_2),  h(e_1,e_2)\rangle+\langle  h(e_2,e_2),  h(e_2,e_2)\rangle+2\delta.
\end{equation}
By considering that $M$ is marginally trapped and using  \eqref{MinkGaussEquation} and \eqref{Pre1Mink4GaussLaplMargTrapped}, we obtain 
\begin{equation}\label{PreMink4GaussLaplMargTrapped}
\|\hat h\|^2=4\delta-2K.
\end{equation}
On the other hand, if $\delta=\pm1$, from $\hat D_{e_i} x=0$, we have
\begin{align}\label{DhatRD}
\begin{split}
 R^{\hat D}(e_1,e_2;\xi,x)=0, &\quad R^{\hat D}(e_1,e_2;\xi,\eta)=R^{D}(e_1,e_2;\xi,\eta),\\
\hat D_{e_i}\hat H=D_{e_i}H&
\end{split}
\end {align}
for all vector fields $\xi,\eta$ tangent to $R^4_1(\delta)$ and normal to $M$. 

We use \eqref{PreMink4GaussLaplMargTrapped} and \eqref{DhatRD} on \eqref{Pre121Mink4GaussLaplMargTrapped} and obtain the following lemma.
\begin{Lemma}\label{Mink4GaussLaplMargTrapped}
Let $M$ be an oriented, partly marginally trapped surface in the space-time $R^4_1(\delta)$ and $\{e_1,e_2;e_3,e_4\}$ an orthonormal frame field on $M$, where $\delta\in\{1,0,-1\}$. Then,  the Laplacian of the Gauss map  $\nu=e_1\wedge e_2$ is
\begin{align}\label{Mink4GaussLaplMargTrppd}
\begin{split}
\Delta\nu=&(4\delta-2K)\nu -2K^De_3\wedge e_4-2D_{e_1}H\wedge e_2-2e_1\wedge D_{e_2}H
\end{split}
\end{align}
where $K$ is the Gaussian curvature of $M$, $K^{D} $ and $H$  are  normal curvature and  mean curvature vector of $M$ in $R^4_1(\delta)$, respectively.
\end{Lemma}

Let $M$ be a marginally trapped surface in the Minkowski space-time and $\{e_1,e_2;e_3,e_4\}$ an orthonormal frame field on $M$. Then, a vector field $C$ defined on $M$ can be expressed as
\begin{subequations}\label{MinkoCVektAll}
\begin{eqnarray}
\label{MinkoCVektTanim}C&=&\sum\limits_{1\leq A<B\leq 4}\varepsilon_A\varepsilon_B C_{AB} e_A \wedge e_B,\\
\label{MinkoCABVektTanim} C_{AB}&=&\langle C, e_{A}\wedge e_{B}\rangle,\; 1\leq A\leq B\leq 4. 
\end{eqnarray}
\end{subequations}
Clearly, $C$ is constant if and only if its components satisfy
\begin{equation}
\label{MinkoCVectConstant}
e_i(C_{AB})=\left\langle C, \left(\widetilde\nabla_{e_i}e_{A}\right)\wedge e_{B}+e_{A}\wedge \widetilde\nabla_{e_i}e_{B}\right\rangle,\quad 1\leq A<B\leq 4,\quad i=1,2.
\end{equation}

Now, we give the following lemma.
\begin{Lemma}\label{PropMargTrap2ndKind}
Let $M$ be a partly marginally trapped surface in the space-time $R^4_1(\delta)$, $\delta\in\{1,0,-1\}$, $H$ its mean curvature vector in $R^4_1(\delta)$ and $\{e_1,e_2;e_3,e_4\}$ an orthonormal frame field on $M$. If a vector field in the form of
\begin{equation}\label{CVectPropMargTrap2ndKind}
C=C_{12}e_1\wedge e_2+C_{34}e_3\wedge e_4+C_{1}e_1\wedge H+C_{2}e_2\wedge H
\end{equation}
is constant, then either $A_H=0$ or $C=0$, where $C_{12},\ C_{34},\ C_1$ and $C_2$ are smooth functions.
\end{Lemma}

\begin{proof}
\textit{Case 1.} $\delta=0$. Without loss of generality, we may assume $\varepsilon_3=-\varepsilon_4=1$ and $H=\alpha(e_3-e_4)$ for a smooth function $\alpha$. We assume that $C$ is a non-zero constant vector. Then,  from  \eqref{MinkoCABVektTanim} and \eqref{CVectPropMargTrap2ndKind} we get $C_{i3}= C_{i4}$ from which we obtain
\begin{subequations}
\begin{eqnarray}
\label{Mink4GaussMargTrppdPW1e34} e_j\left(C_{i3}\right)=e_j\left(C_{i4}\right),\\
\label{Mink4GaussMargTrppdPW1e342} \left\langle C, \left(\nabla_{e_j}e_{i}\right)\wedge (e_3-e_4)\right\rangle=0,\\
\label{Mink4GaussMargTrppdPW1e343} \langle C, e_{i}\wedge D_{e_j}(e_3-e_4)\rangle=0,\quad i,j=1,2.
\end{eqnarray}
\end{subequations}
By a simple calculation using  \eqref{MinkoCVectConstant} and \eqref{Mink4GaussMargTrppdPW1e34}, we get
\begin{equation}\nonumber
\left\langle C,\left(\widetilde\nabla_{e_j}e_{i}\right)\wedge e_3+e_{i}\wedge \widetilde\nabla_{e_j}e_3\right\rangle=\left\langle C,\left(\widetilde\nabla_{e_j}e_{i}\right)\wedge e_4+e_{i}\wedge \widetilde\nabla_{e_j}e_4\right\rangle.
\end{equation}
By using \eqref{MEtomGauss} and \eqref{MEtomWeingarten} on this equation, we obtain
\begin{align}\nonumber
\begin{split}
&\left\langle C, \left(\nabla_{e_j}e_{i}\right)\wedge e_3+h(e_{i},e_j)\wedge e_3-e_{i}\wedge A_3{e_j}+e_{i}\wedge D_{e_j}e_3\right\rangle=\\&\left\langle C,\left(\nabla_{e_j}e_{i}\right)\wedge e_4+h(e_{i},e_j)\wedge e_4-e_{i}\wedge A_4{e_j}+e_{i}\wedge D_{e_j}e_4 \right\rangle.
\end{split}
\end{align}
From this equation, \eqref{MinkoCABVektTanim}, \eqref{Mink4GaussMargTrppdPW1e342} and \eqref{Mink4GaussMargTrppdPW1e343}  we obtain
\begin{equation}\label{LemmasCommonEq01}
h^4_{ij}C_{34}-h^3_{j(3-i)}C_{i(3-i)}=h^3_{ij}C_{34}-h^4_{j(3-i)}C_{i(3-i)}, \ i,j=1,2.
\end{equation}
Therefore, the functions $C_{12}=-C_{21},\ C_{34},\ h^3_{11},\ h^3_{12},\hdots,h^4_{22}$ satisfy
\begin{equation}\label{Mink4GaussMargTrppdPW1AC}
\mathbf A\left(
\begin{array}{c}
C_{12}\\
C_{34}
\end{array}\right)
=0
\end{equation}
for a $4\times2$ matrix $\mathbf A$ given by
$$\mathbf A=\left(
\begin{array}{cc}
h^3_{11}-h^4_{11}&-h^3_{12}+h^4_{12} \\
h^3_{12}-h^4_{12}&h^3_{11}-h^4_{11} \\
h^3_{12}-h^4_{12}&-h^3_{22}+h^4_{22} \\
h^3_{22}-h^4_{22}&h^3_{12}-h^4_{12}
\end{array}\right).$$

Now we will show that $A_3=A_4$. Suppose $A_3\neq A_4$ at a point $p$ of $M$. Then, from \eqref{Mink4GaussMargTrppdPW1AC} we see that there exists a neighborhood $\mathcal N_p$ of $p$ in $M$ on which $C_{34}$ and $C_{12}$ identically vanish. Thus, we have $e_i(C_{12})=e_i(C_{34})=0,\ i=1,2$ from which and \eqref{MinkoCVectConstant} we obtain
\begin{subequations}\label{Mink4GaussMargTrppdPW1ALL}
\begin{eqnarray}
\langle C,-h^3_{i1}e_2\wedge e_3+h^4_{i1}e_2\wedge e_4+h^3_{i2}e_1\wedge e_3-h^4_{i2}e_1\wedge e_4\rangle&=&0,\\
\langle C,-h^3_{i1}e_1\wedge e_4-h^3_{i2}e_2\wedge e_4+h^4_{i1}e_1\wedge e_3+h^4_{i2}e_2\wedge e_3\rangle&=&0
\end{eqnarray}
\end{subequations}
on $\mathcal N_p$. By taking into account $C_{i3}= C_{i4}$ and using \eqref{MinkoCABVektTanim} in \eqref{Mink4GaussMargTrppdPW1ALL}, we obtain
\begin{eqnarray}
\nonumber (h^3_{i2}-h^4_{i2})C_{13}+(-h^3_{i1}+h^4_{i1})C_{23}&=&0,\\
\nonumber (-h^3_{i1}+h^4_{i1})C_{13}+(-h^3_{i2}+h^4_{i2})C_{23}&=&0 ,\quad i=1,2
\end{eqnarray}
on $\mathcal N_p$. As $C\neq0$, from these equation we obtain $(h^3_{i1}-h^4_{i1})^2+(h^3_{i2}-h^4_{i2})^2=0, \ i=1,2$ which leads to  $A_3=A_4$ over $\mathcal N_p$. However, this is a contradiction. Hence, we have $A_3=A_4$ on $M$ which implies $A_H=0$.

\textit{Case 2}. $\delta=\pm1$. Let $x$ be the position vector of $M$, $C$ a mapping given by \eqref{CVectPropMargTrap2ndKind}. Then, \eqref{CVectPropMargTrap2ndKind} implies $\langle C, e_A\wedge x\rangle=0$ from which we obtain
\begin{align}
\begin{split} 
e_i\big(\langle C,e_A\wedge e_B\rangle\big)=&\langle e_i(C),e_A\wedge e_B\rangle+\left\langle C,\left(\hat\nabla_{e_i}e_A\right)\wedge e_B+e_A\wedge \hat\nabla_{e_i}e_B\right\rangle\\
=&\langle e_i(C),e_A\wedge e_B\rangle+\left\langle\left(\widetilde\nabla_{e_i}e_A\right)\wedge e_B+e_A\wedge \widetilde\nabla_{e_i}e_B\right\rangle,
\end{split}
\end{align}
where $\hat\nabla$ denotes the Levi-Civita connection of $E(4,1,\delta)$, $A,B=1,2,3,4$.
Therefore, $C$ is constant if and only if \eqref{MinkoCVectConstant} is satisfied. By the exactly same way with \textit{Case 1}, we obtain either $C=0$ or $A_H=0$.
\end{proof}
The following Lemma is obtained from the proof of Theorem 6.1 and Theorem 8.1 in \cite{ChenVeken2009Houston}.
\begin{Lemma}\cite{ChenVeken2009Houston}\label{LemmaParallelMeanCurvature}
Let $M$ be a partly marginally trapped surface in $R^4_1(\delta)$ for $\delta=\pm1$. Assume that mean curvature vector $H$ of  $M$ in $R^4_1(\delta)$ is parallel. Then, there exists a local orthonormal base field $\{e_1,e_2\}$ and a normal light-like vector field $f_4$ such that $h(e_1,e_1)=(1-a_1)H+a_2f_4$, $h(e_1,e_2)=0$, $h(e_2,e_2)=(1+a_1)H-a_2f_4$, $\langle H,f_4\rangle=-1$ for some smooth functions $a_1$ and $a_2$.
\end{Lemma}

\section{Marginally trapped surfaces in the Minkowski space-time}
\subsection{Pointwise 1-type Gauss map of the first kind}
It is well-known that for a non-maximal space-like surface $M$ in $\mathbb E^4_1$, being parallel of $H$ implies $K^D\equiv0$ (see, for example, \cite[Lemma 3.3]{Chen2009JMP}). Thus, the next proposition directly follows from Lemma \ref{Mink4GaussLaplMargTrapped}.
\begin{Prop}\label{DursunTurgayE41SpacelikeTHM}
Let $M$ be a marginally trapped surface in the Minkowski space-time. Then, $M$ has pointwise 1-type Gauss map of the first kind if and only if  $M$   has parallel mean curvature vector. Moreover, $M$ has harmonic Gauss map if and only if  $M$ is a flat surface with parallel mean curvature vector.
\end{Prop}
\begin{Remark}
See \cite{ChenVeken2009Houston}, for the classification of flat marginally trapped surfaces in the Minkowski space-time with parallel mean curvature vector.
\end{Remark}
\begin{theorem}\label{QGAClassificationThmE41Pw1stkind}
Let $M$ be a marginally trapped surface in the Minkowski space-time. Then, $M$ has \textbf{proper} pointwise 1-type Gauss map of the first kind if and only if it is a non-flat surface lying in $\mathbb S^3_1(r^2)$ or $\mathbb H^3(-r^2)$.
\end{theorem}

\begin{proof}
Let $M$ have proper pointwise 1-type Gauss map of the first kind. Then, $M$ is non-flat and it has parallel mean curvature vector, because of {Proposition }\ref{DursunTurgayE41SpacelikeTHM}. According to \cite[Theorem 4.1]{ChenVeken2009Houston}, a non-flat marginally trapped surface with parallel mean curvature vector is lying in $\mathbb S^3_1(r^2)$ or $\mathbb H^3(-r^2)$.

Conversely, let $M$ be a non-flat marginally trapped surface lying in $\mathbb S^3_1(r^2)$ or $\mathbb H^3(-r^2)$ and $x$ its position vector in the Minkowski space-time. Consider the orthonormal frame field $\{e_3,e_4\}$ of normal bundle  of $M$ such that $e_3=rx$ and $H=-\varepsilon_3 r(e_3-e_4).$ As $\widetilde\nabla_{e_i}x=e_i$, $e_3$ is parallel. Since $M$ has codimension 2, $De_4=0$. Therefore, we have $H$ is parallel. Thus, the Laplacian of the Gauss map $\nu$ of $M$ becomes $\Delta\nu=-2K\nu$ because of Lemma \ref{Mink4GaussLaplMargTrapped}. 

Now, we will show that $K$ is not constant. Consider an orthonormal base field $\{e_1,e_2\}$ of tangent bundle of $M$ such that the corresponding shape operators are 
\begin{equation}\label{LabelPROPER1}
A_3=-rI,\quad A_4=\mathrm{diag}(-r+\zeta,-r-\zeta)
\end{equation}
for a smooth function $\zeta$. Then, the Gauss equation \eqref{MinkGaussEquation} and \eqref{LabelPROPER1} imply $K=\varepsilon_3r^2\zeta^2$. 
We assume that $K$  is constant, or equivalently, $\zeta$ is constant. By a direct calculation using the Codazzi equation \eqref{MinkCodazzi}, we obtain $2\omega_{12}(e_i)\zeta=0,\ i=1,2$ which implies $M$ is flat. However, this is a contradiction. Hence, $K$ is not constant which yields that $M$ has \textit{proper}  pointwise 1-type Gauss map of the first kind.
\end{proof}
Moreover, we have
\begin{Corol}
Let $M$ be a partly marginally trapped surface in the Minkowski space-time with Gauss map $\nu$  satisfying \eqref{Glbl1TypeDefinition} for $C=0$. Then $\nu$ is harmonic.
\end{Corol}
\subsection{Gauss map of pseudo-umbilical marginally trapped surfaces in $\mathbb E^4_1$}
A submanifold in  $\mathbb E^m_s$ is said to be pseudo-umbilical if there exists a function $\rho$ such that second fundamental form $h$ and  mean curvature vector $H$ of $M$ satisfy
\begin{equation}\label{PSEUMBDEFINITON}
\langle h(X,Y),H\rangle=\rho\langle X,Y \rangle
\end{equation}
for all tangent vector fields $X,Y$ of $M$. In this case, \eqref{PSEUMBDEFINITON} is satisfied for the function $\rho$   given by $\rho=\langle H,H \rangle$, \cite{CabrerizoComplete}. Thus, from \eqref{MinkAhhRelatedby} and \eqref{PSEUMBDEFINITON} one can see that a marginally trapped surface in $\mathbb E^4_1$ is pseudo-umbilical if and only if  $A_H=0$.

\begin{Remark}\label{RemarkONChenIshikawa1991Bih}
Let $M$ be a (partly) marginally trapped surface in $\mathbb E^4_1$. If  $A_H\equiv0$ on $M$, then $M$ is congruent to the surface given by 
\begin{equation}\label{MinkMargTrapnottooComplete}
x(u,v) = (\phi (u, v), u, v, \phi (u, v)),
\end{equation}
for a smooth function $\phi:\Omega\rightarrow\mathbb R$, where $\Omega$ is an open subset of $\mathbb R^2$ (see the proof of Theorem 6.1 in \cite{ChenIshikawa1991Bih} and also \cite{CabrerizoComplete}). 
\end{Remark}
Let $M$ be the surface in $\mathbb E^4_1$ given by \eqref{MinkMargTrapnottooComplete} for a smooth function $\phi:\Omega\rightarrow\mathbb R$. We consider the local orthonormal base field $\{e_1,e_2\}$ of tangent bundle of $M$ given by
\begin{equation}
\label{MinkMargTrapnottooCompletee1} e_1=\frac{\partial}{\partial u}=(\phi_u,1,0,\phi_u),\quad e_2=\frac{\partial}{\partial v}=(\phi_v,0,1,\phi_v).
\end{equation}
By a direct computation, we obtain $\omega_{12}\equiv0$ and $K=K^D\equiv0$ on $M$. Moreover,  mean curvature vector $H$ becomes 
\begin{equation}\label{MinkMarge3e4Eq}
H=-\frac{\Delta\phi} 2 (1,0,0,1).
\end{equation}
Thus, from \eqref{Mink4GaussLaplMargTrppd} for $\delta=0$ we obtain
\begin{equation}\label{MinkMargTrapnottooCompleteMapDnu}
\Delta\nu=(1,0,0,1)\wedge\left(-{\Delta\phi_v}  e_1+{\Delta\phi_u} e_2\right).
\end{equation}
On the other hand, from \eqref{MinkMargTrapnottooCompletee1} we have
\begin{equation}\label{MinkMargTrapnottooCompleteMapnu}
\nu=(1,0,0,1)\wedge\left(-\phi_ve_1+\phi_ue_2\right)+(0,1,0,0)\wedge (0,0,1,0).
\end{equation}
In the next proposition, we obtain a family of pseudo-umbilical surfaces with  pointwise 1-type Gauss map of the second kind.
\begin{Prop}\label{PseudoUmbSurf2ndkind}
Let $\Omega$ be an open, bounded subset of $\mathbb R^2$, $\psi,\phi:\Omega\rightarrow \mathbb R$ some smooth functions satisfying 
\begin{subequations}\label{PDEforPHI}
\begin{eqnarray}
\label{PDEforPHI001} \Delta\psi=F(\psi),\\
\label{PDEforPHI002} \phi(u,v)=\psi(u,v)+c_1u+c_2v
\end{eqnarray}
\end{subequations}
for some constants $c_1,\ c_2$ and a differentiable non-constant function $F:\psi(\Omega)\rightarrow\mathbb R$ such that the function $f:\Omega\rightarrow\mathbb R$ defined by 
\begin{eqnarray}
\label{PDEforPHIFuncf} f(u,v)&=&F'(\psi(u,v))
\end{eqnarray}
is smooth. Consider the partly marginally trapped surface $M$ in the Minkowski space-time given by \eqref{MinkMargTrapnottooComplete}.

Then, the Gauss map of $M$ is pointwise 1-type of the second kind and satisfies \eqref{PW1TypeDefinition} for the smooth function $f$ given in \eqref{PDEforPHIFuncf} and the constant vector
\begin{eqnarray}
\label{PDEforPHIConstC} C&=&(1,0,0,1)\wedge (0,-c_2,c_1,0)-(0,1,0,0)\wedge(0,0,1,0).
\end{eqnarray}
\end{Prop}
\begin{proof}
Let $e_1,\ e_2$ be the vector fields given in \eqref{MinkMargTrapnottooCompletee1}. If $\phi$ satisfies \eqref{PDEforPHI}, then we have $\Delta\phi_u=(\phi_u+c_1)F'(\psi)$ and $\Delta\phi_v=(\phi_v+c_2)F'(\psi)$. These equations, \eqref{MinkMargTrapnottooCompleteMapDnu} and \eqref{PDEforPHIFuncf} imply
$\Delta\nu=f(1,0,0,1)\wedge\Big(-\phi_ve_1+\phi_ue_2-c_2e_1+c_1e_2\Big)$
from which and \eqref{MinkMargTrapnottooCompleteMapnu} we obtain
\begin{equation}\nonumber
\Delta\nu=f\Big(e_1\wedge e_2-(0,1,0,0)\wedge(0,0,1,0)+(1,0,0,1)\wedge(-c_2e_1+c_1e_2)\Big).
\end{equation}
From this equation, \eqref{MinkMargTrapnottooCompletee1} and \eqref{PDEforPHIConstC} we see that \eqref{PW1TypeDefinition} is satisfied for the constant vector $C$ and the smooth function  $f\not\equiv 0$ given in the proposition.  Hence, $M$ has pointwise 1-type Gauss map of the second kind.
\end{proof}

Next, we obtain the following classification theorem of pseudo-umbilical marginally trapped surfaces in $\mathbb E^4_1$ with pointwise 1-type Gauss map.
\begin{Prop}\label{PropMargTrap2ndKindExampleProp}
Let $M$ be a pseudo-umbilical marginally trapped surface in the Minkowski space-time. Then, $M$ has pointwise 1-type Gauss map of the second kind if and only if $M$ is congruent to a surface given in Proposition \ref{PseudoUmbSurf2ndkind}.
\end{Prop}

\begin{proof}
Let $M$ be a marginally trapped surface in the Minkowski space-time given by \eqref{MinkMargTrapnottooComplete}.  Then, the Gauss map $\nu$ of $M$ satisfies \eqref{MinkMargTrapnottooCompleteMapDnu} and \eqref{MinkMargTrapnottooCompleteMapnu}, where $e_1,\ e_2$ are the vector fields given by \eqref{MinkMargTrapnottooCompletee1}.

Now, in order to prove the necessary condition, we assume that $M$ has pointwise 1-type Gauss map of the second kind. Then, the equation \eqref{PW1TypeDefinition} is satisfied for a smooth function $f\not\equiv 0$ and a constant vector $C\neq0$. 

If $f\equiv0$ on an open, connected subset $\mathcal O$ of $M$, then we have $(\Delta\nu)\big|_\mathcal O\equiv 0$ and \eqref{MinkMargTrapnottooCompleteMapDnu} implies $(\Delta\phi)\big|_\mathcal O\equiv c_\mathcal O$ for a non-zero constant $c_\mathcal O\in\mathbb R$. Thus, \eqref{PDEforPHI} is satisfied for $F=c_\mathcal O$ on $\mathcal O$. Now we assume that $f$ is non-vanishing on $M$.

From \eqref{PW1TypeDefinition}, \eqref{MinkMargTrapnottooCompleteMapDnu} and \eqref{MinkMargTrapnottooCompleteMapnu} we obtain
\begin{equation}\label{MinkMargTrapnottooCompleteMapnuzeta}
(1,0,0,1)\wedge\zeta=C+(0,1,0,0)\wedge (0,0,1,0),
\end{equation}
where $\zeta$ is the tangent vector field given by 
\begin{equation}\label{MinkMargTrapnottooCompleteMapnuzeta2}
\zeta= \left(-\frac{\Delta\phi_v} {f}+\phi_v\right) e_1+\left(\frac{\Delta\phi_u} {f}-\phi_u\right) e_2.
\end{equation}
By using \eqref{MinkMargTrapnottooCompleteMapnuzeta}, we obtain 
$$
(1,0,0,1)\wedge\widetilde\nabla_{e_i}\zeta=e_i\big((1,0,0,1)\wedge\zeta\big)=e_i\big(C+(0,1,0,0)\wedge (0,0,1,0)\big)=0,\ i=1,2$$ 
from which we see that $\widetilde\nabla_{e_i}\zeta$ is proportional to $(1,0,0,1)$ which is normal to $M$ because of \eqref{MinkMarge3e4Eq}. Thus, $\nabla_{e_i}\zeta=0$  and \eqref{MinkMargTrapnottooCompleteMapnuzeta2} implies
$$\nabla_{e_i}\left(\left(-\frac{\Delta\phi_v} {f}+\phi_v\right) e_1+\left(\frac{\Delta\phi_u} {f}-\phi_u\right)e_2\right)=0,\quad i=1,2.$$
As $\omega_{12}=0$, these equations imply 
\begin{equation}\label{MinkMargTrapnottooCompletePW1GYKos}
\frac{\Delta\phi_u} {f}-\phi_u=c_1 \quad \mathrm{and}\quad \frac{\Delta\phi_v} {f}-\phi_v=c_2 
\end{equation}
for some constants $c_1$ and $c_2$ from which we have $(\phi_v+c_2)\Delta\psi_u=(\phi_u+c_1)\Delta\psi_v$ where $\psi$ is the function given by \eqref{PDEforPHI001}. By solving this equation, we obtain \eqref{PDEforPHI002} for a differentiable function $F$. 

Converse is given in Proposition \ref{PseudoUmbSurf2ndkind}.
\end{proof}

\subsection{Pointwise 1-type Gauss map of the second kind}
By using  Lemma \ref{PropMargTrap2ndKind}, Proposition \ref{PropMargTrap2ndKindExampleProp} and Remark \ref{RemarkONChenIshikawa1991Bih}, we state the following classification theorem.
\begin{theorem}\label{MargTrppdSrfsClassThm2ndKIND}
Let $M$ be a marginally trapped surface in the Minkowski space-time. Then, $M$ has pointwise 1-type Gauss map of the second kind if and only if it is congruent to a surface given in Proposition \ref{PseudoUmbSurf2ndkind}.
\end{theorem}
\begin{proof}
Let $M$  have pointwise 1-type Gauss map of the second kind. Then, its Gauss map satisfies \eqref{PW1TypeDefinition} for a smooth function $f\not\equiv0$ and a constant vector field $C\neq0$. From \eqref{PW1TypeDefinition} and \eqref{Mink4GaussLaplMargTrppd} for $\delta=0$ we see that $C$ is of the form of \eqref{CVectPropMargTrap2ndKind}. Since $C$ is constant, Lemma \ref{PropMargTrap2ndKind} implies $A_H=0$. Thus, $M$ is congruent to a surface given in Proposition \ref{PseudoUmbSurf2ndkind}.

Converse is given in Proposition \ref{PseudoUmbSurf2ndkind}.
\end{proof}
Next, we want to give following corollary obtained from Lemma \ref{PropMargTrap2ndKind} and \cite[Theorem 4.3]{Chen2009JMP}.
\begin{Corol}
There are no marginally trapped surface with pointwise 1-type Gauss map of the second kind in the Minkowski space-time lying in the light cone ${\mathcal{LC}}^3$.
\end{Corol}
\subsection{Classification of marginally trapped surfaces with pointwise 1-type Gauss map}
By combaining Theorem \ref{QGAClassificationThmE41Pw1stkind} and Theorem \ref{MargTrppdSrfsClassThm2ndKIND}, we have the following classification theorem:
\begin{theorem}\label{QGAClassificationThm}
Let $M$ be a marginally trapped surface in the Minkowski space-time and $\mathcal M$ the open subset of $M$ given by $\mathcal M=\{p\in M|\Delta\nu|_p\neq0\}$. Then, $M$ has pointwise 1-type Gauss map if and only if $\mathcal M$ is congruent to  one of the following two type of surfaces. 
\begin{enumerate}
\item[(i)] a non-flat surface  lying in the de Sitter space-time $S^3_1(r^2)$ or $H^3(-r^2)$ for $r>0$.
\item[(ii)] a surface congruent to a surface given in Proposition \ref{PseudoUmbSurf2ndkind}. 
\end{enumerate}
\end{theorem}
By combaining Proposition \ref{PseudoUmbSurf2ndkind} and Theorem \ref{QGAClassificationThm}, we obtain classification of marginally trapped surfaces with 1-type Gauss map.
\begin{theorem}
Let $M$ be a marginally trapped surface in the Minkowski space-time. Then,  $M$ has non-harmonic 1-type Gauss map, if and only if it is congruent to the surface given by \eqref{MinkMargTrapnottooComplete} for a smooth function $\phi:\Omega\rightarrow \mathbb R$ satisfying
Helmholtz equation 
\begin{equation}\label{EqPHIPW1TYPE2ndGlbl}
\Delta\phi+\lambda\phi=c_1u+c_2v
\end{equation}
for some constants $\lambda\neq0,\ c_1,\ c_2$.
\end{theorem}
\begin{Corol}
If  a marginally trapped surface $M$ has non-harmonic, 1-type Gauss map, then it is null 2-type.
\end{Corol}
\begin{proof}
We assume that $M$ is the surface given by \eqref{MinkMargTrapnottooComplete} for a smooth function $\phi$ satisfying \eqref{EqPHIPW1TYPE2ndGlbl}. We define the mappings $x_0$ and $x_1$ by $x_0(u,v)=\Big(-\frac{c_1u+c_2v}\lambda ,u,v,-\frac{c_1u+c_2v}\lambda\Big)$ and $x_1(u,v)=\Big(\phi+\frac{c_1u+c_2v}\lambda,0,0,\phi+\frac{c_1u+c_2v}\lambda\Big).$
Then, we see that these mappings satisfy $x=x_0+x_1$, $\Delta x_0=0$ and $\Delta x_1=\lambda x_1$. Hence $M$ is null 2-type.
\end{proof}

\subsection{Examples}
In this subsection, we present explicit examples of (partly) marginally trapped surfaces with pointwise 1-type Gauss map. We also give construction of marginally trapped surfaces with 1-type Gauss map for a given boundary curve.

In the first example, we construct a marginally trapped surface with proper pointwise 1-type Gauss map of the second kind. 
\begin{Example}
Let $\epsilon>0$ and $\Omega_\epsilon=\{(u,v)|\epsilon<u+v<1\}$ and $\phi:\Omega\rightarrow \mathbb R$ be a function given by 
\begin{equation}\nonumber
\phi(u,v)= \exp\left(\frac{1}{u+v}\right) 
\end{equation}
Then, by a direct calculation, we see that $\phi:\Omega_\epsilon\rightarrow\mathbb R$ satisfies \eqref{PDEforPHI} for the differentiable function $F(\phi)= \big(-2(\ln\phi)^4-4(\ln\phi)^3\big)\phi$ and the function given by \eqref{PDEforPHIFuncf} becomes
\begin{equation}\label{PseudoUmbSurf2ndkindEx2f}
f(u,v)=-\frac{12}{(u+v)^2}-\frac{12}{(u+v)^3}-\frac{2}{(u+v)^4}
\end{equation}
which is smooth on $\Omega_\epsilon$. Hence, Proposition \ref{PseudoUmbSurf2ndkind} implies that the marginally trapped surface $ M$ given by \eqref{MinkMargTrapnottooComplete} has pointwise 1-type Gauss map of the second kind.
\end{Example}

In the next proposition, we  obtain construction of marginally trapped surfaces in the Minkowski space-time with 1-type Gauss map for a given boundary curve.
\begin{Prop}\label{PropJordanCurve}
Let $\beta$ be a continuous Jordan curve lying in a space-like plane of the Minkowski space-time.  Then, for each eigenvalue $\lambda_k$ of operator $\Delta$ given in \eqref{OpDeltaOmega} there exists a (partly) marginally trapped surface $M_k$ with the boundary $\beta$ and 1-type Gauss map satisfying \eqref{PW1TypeDefinition} for $f=\lambda_k$ and $C=(0,1,0,0)\wedge (0,0,1,0).$
\end{Prop}
\begin{proof}
Up to isometries of the Minkowski space-time, we assume that $\beta$ is contained by the space-like plane $\Pi=\{(0,u,v,0)|u,v\in\mathbb R\}$. Let $\Omega$ be the open subset of $\Pi$ bounded by $\beta$. As $\beta$ is a Jordan curve, $\Omega$ is connected. Consider a solution $\phi_k$ of the boundary value problem given in \eqref{BVPHelmoltz} for an eigenvalue $\lambda_k$ of \eqref{OpDeltaOmega}. Then, $\phi_k$ is smooth (see Section \ref{subsectionHelmoltz}) and satisfies \eqref{EqPHIPW1TYPE2ndGlbl} for $c_1=c_2=0$. 

Let $M_k$ be (partly) marginally trapped surface in the Minkowski space-time given by \eqref{MinkMargTrapnottooComplete} for $\phi=\phi_k$. Then, $\phi$ satisfies \eqref{PDEforPHI} for  the function $F=\lambda_k\phi$. Moreover, the function $f$ given by \eqref{PDEforPHI} becomes a constant function. Thus, it is smooth. Hence, $M_k$ has 1-type Gauss map because of Proposition \ref{PseudoUmbSurf2ndkind}. Moreover, as $\phi_k=0$ on $\partial\Omega=\beta$, $\beta$ is the boundary curve of $M_k$.
\end{proof}
In the next example, we construct a marginally trapped surface with the boundary of a square lying in the space-like plane $\Pi$.
\begin{Example}\label{ExampleGivenBoundary}
The function $\phi(u,v)=\sin (n\pi x )sin (n\pi y)$ solves the boundary value problem \eqref{BVPHelmoltz} for $\Omega=(0,1)\times(0,1)$, where $n\in\mathbb N$. Therefore, the surface given by $x=(\phi(u,v),u,v,\phi(u,v))$ has 1-type Gauss map. Moreover, the boundary curve of this surface is the non-smooth curve $\beta$ which is a square in the plane $\Pi$.
 \end{Example}

\section{Marginally trapped surfaces in the de Sitter and anti-de Sitter space-times}
In this section, we obtain marginally trapped surfaces in the de-Sitter space-time $\mathbb S^4_1(1)$ and anti-de Sitter space-time $\mathbb H^4_1(-1)$. 

First, we want to prove the following proposition.
\begin{Prop}\label{THMAH0DeSitter}
Let $M$ be a marginally trapped surface in the space-time $R^4_1(\delta)$ $\delta=\pm1$, and $H$ its mean curvature vector. If $A_H=0$, then $H$ is parallel.
\end{Prop}
\begin{proof}
Let $x=(x_0,x_1,x_2,x_3,x_4):\Omega\rightarrow E(4,1,\delta)$ be the position vector of $M$ and $u,v$ local coordinates on $M$ such that the induced metric of $M$  is $g=m^2(du^2+dv^2)$ for a non-vanishing smooth function $m$.  Then, the Laplace operator becomes $\Delta=m^{-2}(-\partial_u^2-\partial_v^2)$. 
Consider the orthonormal frame field $\{e_1,e_2;e_3,e_4\}$ with $e_1=m^{-1}\partial_u$, $e_2=m^{-1}\partial_v$ and
\begin{equation}\label{THMPW1DeSitterH}
H=\alpha(e_3-e_4).
\end{equation}
Since $A_H=0$, we have $A_3=A_4$.  From Ricci equation \eqref{MinkRicciEquation} we have $R^{D}=0$. Therefore, we may choose $e_3,\ e_4$ such that $\omega_{34}=0$. Then, we have $A_{e_3-e_4}=0$ and $\hat D(e_3-e_4)=0$, where $\hat D$ is normal connection of $M$ in $E(4,1,\delta)$. Thus, $e_3-e_4$ is a constant light-like vector. Up to linear isometries of $E(4,1,\delta)$, we may assume 
\begin{equation}\label{THMPW1DeSittere34}
e_3-e_4=(a,0,0,0,a)
\end{equation}
for a constant $a\neq0$. From the Laplace-Beltrami equation $\Delta x=-2H$,  \eqref{THMPW1DeSitterH} and \eqref{THMPW1DeSittere34} we obtain $\Delta x=-2\alpha(a,0,0,0,a)$
which implies
\begin{equation}\label{THMPW1DeSittereLaplx}
\Delta x_i=0,\quad i=1,2,3.
\end{equation}

On the other hand, as $\langle e_1,e_3-4\rangle=\langle e_2,e_3-4\rangle=0$, we have $x_4=x_0+c$ for a constant $c$ from \eqref{THMPW1DeSittere34}. As $g=m^2(du^2+dv^2)$, we have
\begin{subequations}\label{THMPW1DeSitterexuxv}
\begin{eqnarray}
\delta\left(\frac{\partial x_1}{\partial u}\right)^2+\left(\frac{\partial x_2}{\partial u}\right)^2+\left(\frac{\partial x_3}{\partial u}\right)^2&=&m^2,\\
\delta\left(\frac{\partial x_1}{\partial v}\right)^2+\left(\frac{\partial x_2}{\partial v}\right)^2+\left(\frac{\partial x_3}{\partial v}\right)^2&=&m^2,
\end{eqnarray}
\end{subequations}
 From $\langle x,x\rangle=\delta$ we obtain $c^2+2cx_0+\delta x_1^2+x_2^2+x_3^2=\delta$. By taking Laplacian of both sides of this equation, we obtain
\begin{align}\label{THMPW1DeSittereLaplAll}
\begin{split}
&2c\Delta x_0+2\delta x_1\Delta x_1+2x_2\Delta x_2+2x_3\Delta x_3-\frac{2}{m^2}\left(\delta\left(\frac{\partial x_1}{\partial u}\right)^2+\left(\frac{\partial x_2}{\partial u}\right)^2+\left(\frac{\partial x_3}{\partial u}\right)^2\right.\\&+\left.\delta\left(\frac{\partial x_1}{\partial v}\right)^2+\left(\frac{\partial x_2}{\partial v}\right)^2+\left(\frac{\partial x_3}{\partial v}\right)^2\right)=0.
\end{split}
\end{align}
By combaining \eqref{THMPW1DeSittereLaplx}-\eqref{THMPW1DeSittereLaplAll}, we obtain
$\Delta x_0=\Delta x_4=\mbox{const}.$ From this equation and \eqref{THMPW1DeSittereLaplx} we see that $\Delta x$ is a constant vector. Thus, the Laplace-Beltrami equation implies that $H$ is constant.  Hence, $H$ is parallel.
\end{proof}
Next, we obtain the following theorem.
\begin{theorem}\label{THMPW1DeSitter}
Let $M$ be a marginally trapped surface in the  space-time  $R^4_1(\delta)$, $\delta=\pm1$. Then, the following statements are logically equivalent.
\begin{enumerate}
\item[(i)]    $M$ has pointwise 1-type Gauss map;
\item[(ii)]   $M$ has pointwise 1-type Gauss map of the first kind;
\item[(iii)]  $M$ has parallel mean curvature vector.
\end{enumerate}
\end{theorem}
\begin{proof}
First, we  prove (i)$\Rightarrow$(iii). Let $M$ have  pointwise 1-type Gauss map. Then, its Gauss map satisfies \eqref{PW1TypeDefinition} for a smooth function $f$ and a constant vector field $C$. If $C=0$, then \eqref{Mink4GaussLaplMargTrppd} implies $D_{e_i}H=0,\ i=1,2$. Thus, the proof is completed. 

We assume $C\neq0$. Then, from \eqref{PW1TypeDefinition} and \eqref{Mink4GaussLaplMargTrppd} for $\delta=1$ we see that $C$ is of the form of \eqref{CVectPropMargTrap2ndKind}. Since $C$ is constant, Lemma \ref{PropMargTrap2ndKind} implies $A_H=0$ and from Proposition \ref{THMAH0DeSitter} we see $DH=0$. Hence, $H$ is parallel.

(iii)$\Rightarrow$(ii) directly follows from Lemma \ref{Mink4GaussLaplMargTrapped} as $DH=0$ implies $K^D=0$. (ii)$\Rightarrow$(i) is obvious.
\end{proof}
\begin{Remark}
See \cite{ChenVeken2009Houston}, for the classification of marginally trapped surfaces in the de Sitter and anti de-Sitter space-time with parallel mean curvature vector.
\end{Remark}
\begin{Prop}\label{Glbl1TypePROP}
Let $M$ be a marginally trapped surface in the  space-time $R^4_1(\delta),\ \delta=\pm1$. If $M$ has 1-type Gauss map, then $\Delta\nu=4\delta\nu$ or $\Delta\nu=2\delta\nu$.
\end{Prop}
\begin{proof}
If $M$ has pointwise 1-type Gauss map, then  Theorem \ref{THMPW1DeSitter} implies $M$ has parallel mean curvature veector.  Let $e_1,e_2,f_4$ be vector fields given in Lemma \ref{LemmaParallelMeanCurvature} for some smooth functions $a_1,\ a_2$. Then, from the Gauss equation \eqref{MinkGaussEquation} we have $K=2a_1a_2+\delta$. Thus,  from \eqref{Mink4GaussLaplMargTrppd} for $\delta=\pm1$ we obtain 
\begin{equation}\label{ParallelMeanCurvaturenu}
\Delta\nu=(2\delta-4a_1a_2)\nu.
\end{equation}

On the other hand, from the Codazzi equation \eqref{MinkCodazzi} we get 
\begin{equation}\label{ParallelMeanCurvatureCodazzi}
e_1(a_i)=-2a_i\omega_{12}(e_2),\quad e_2(a_i)=2a_i\omega_{12}(e_1), \quad i=1,2.
\end{equation}
Now, assume that $M$  has  1-type Gauss map. Then, we have $a_1a_2=\mbox{const}$ from which and \eqref{ParallelMeanCurvatureCodazzi} we see $a_1$ and $a_2$ are constants. Moreover, from \eqref{ParallelMeanCurvatureCodazzi} we see that either $\omega_{12}=0$ or $a_1=a_2=0$. Because of \eqref{ParallelMeanCurvaturenu}, these two cases imply either $\Delta\nu=4\delta\nu$ or $\Delta\nu=2\delta\nu$, respectively.
\end{proof}
We state the following corollary of Proposition \ref{Glbl1TypePROP}.
\begin{Corol}
There is no marginally trapped surface in the de Sitter or anti-de Sitter space-times with harmonic Gauss map.
\end{Corol}
Next, we want to give an explicit example of marginally trapped surface in the de Sitter space-time with 1-type Gauss map.
\begin{Example}
Let $M$ be the marginally trapped surface in the de Sitter space-time given by $x(u,v)=(1,\sin u,\cos u\cos v,\cos u\sin v,1)$. Then,  mean curvature vector of $M$ in $\mathbb S^4_1(1)$ is parallel and its Gaussian curvature is 1, \cite{ChenVeken2009Houston}. Thus, \eqref{Mink4GaussLaplMargTrppd} for $\delta=1$ implies $\Delta\nu=2\nu$. Thus, $M$ has 1-type Gauss map. 
\end{Example}


\end{document}